\documentclass[preprint]{revtex4-1}

\usepackage{amsmath,amsthm,amssymb}
\usepackage[colorlinks=true,allcolors=blue]{hyperref}
\usepackage{pstricks,multido,pstricks-add}
\usepackage[nomessages]{fp}

\newtheorem{thm}{Theorem}
\newtheorem{prop}[thm]{Proposition}
\newtheorem{lemma}{Lemma}
\numberwithin{equation}{section}

\def \C {\mathbb{C}}
\def \R {\mathbb{R}}
\def \Z {\mathbb{Z}}

\def \H_g so(N){H_g^{so(N)}}
\def \Hggl(N){H_g^{gl(N)}}

\newcommand{\nad}[2]{\genfrac{}{}{0pt}{}{#1}{#2}}
\def \p {\partial}

\renewcommand{\a}{\alpha}
\def\beq#1#2\eeq{%
        \begin{equation}%
        \label{#1}%
            #2%
        \end{equation}%
    }
\def\plaquette#1#2#3{
    \def\xa{#1} \def\ya{#2}
    \FPadd{\xb}{#3}{\xa} \FPadd{\yb}{#3}{\ya}
    \pnode(\xa,\ya){i}\pnode(\xa,\yb){j}\pnode(\xb,\ya){k}\pnode(\xb,\yb){l}
    \psdots[linecolor=blue](\xa,\ya)(\xb,\ya)(\xa,\yb)(\xb,\yb)
    \uput[-135](\xa,\ya){$i$}\uput[135](\xa,\yb){$j$}
    \uput[-45](\xb,\ya){$k$}\uput[45](\xb,\yb){$l$}
}

\def\fourdotsijkl#1#2#3{
  \def\xa{#1}\def\dx{#2}\def\x{\xa}
   \multido{\inode=1+1,\nd=\xa+\dx}{4} {
        \psdot[linecolor=blue](\nd,#3)
        \pnode(\nd,#3){x\inode}
    }
    \nput{-90}{x1}{$i$}\nput{-90}{x2}{$j$}\nput{-90}{x3}{$l$}\nput{-90}{x4}{$k$}
       \ncline[linestyle=dotted,arrows=-,ArrowInsideNo=0]{x1}{x4}
}

\begin{document}

\title{On Dunkl angular momenta algebra}

\author{Misha Feigin}
\email{misha.feigin@glasgow.ac.uk}
\affiliation{School of Mathematics and Statistics,
University of Glasgow,
15 University Gardens,
Glasgow G12 8QW, UK}
\author{Tigran Hakobyan}
\email{tigran.hakobyan@ysu.am}
\affiliation{Yerevan State University, 1 Alex Manoogian, 0025 Yerevan, Armenia}
\affiliation{Tomsk Polytechnic University, Lenin Ave. 30, 634050 Tomsk, Russia}

\begin{abstract}
We consider the quantum angular momentum generators, deformed by means of the Dunkl operators. Together with the reflection operators they generate a subalgebra in the rational Cherednik algebra associated with a finite real reflection group. We find all the defining relations of the algebra, which appear to be quadratic, and we show that the algebra is of Poincar\'e-Birkhoff-Witt (PBW) type. We show that this algebra contains the angular part of the Calogero--Moser Hamiltonian and that together with constants it generates the centre of the algebra. We also consider the $gl(N)$ version of the subalgebra of the rational Cherednik algebra and show that it is a non-homogeneous quadratic algebra of PBW type as well. In this case the central generator can be identified with the usual Calogero--Moser Hamiltonian associated with the Coxeter group in the harmonic confinement.


\end{abstract}
\maketitle

\section{Introduction}
The  differential-difference operators
\begin{equation}
\label{Dunkl}
{\mathcal D}_i=\partial_{x_i}+\sum_{\a \in {\mathcal R}_+}\frac{g_\alpha  \a_i}{(\alpha, x)}(1-s_{\alpha})
\end{equation}
were introduced by Dunkl  in
 \cite{dunkl89} in the context of study of so called  Dunkl-harmonic and multivariable orthogonal polynomials. Here
${\mathcal R}_+$ is a positive half of a Coxeter root system $\mathcal R$, and $s_\alpha$ are the corresponding orthogonal reflections from the associated Coxeter group $W$, $\a=(\a_1,\ldots,\a_N)$. In the case of ${\mathcal R}= A_{N-1}$ they take the form of the permutation operators $s_{ij}$ which act on the
space of functions as follows:
\begin{equation*}
(s_{ij}\psi)(x_1,\dots,x_i,\dots,x_j,\dots x_N)=\psi(x_1,\dots,x_j,\dots,x_i,\dots x_N).
\end{equation*}
The function $g(\alpha)=g_\alpha$ is $W$-invariant so in the case of type $A_{N-1}$ it is an arbitrary 
complex constant. The operators ${\mathcal D}_i$ may be considered as $g$-dependent deformations of the partial differentiation operators
$\partial_i=\partial_{x_i}$.

These operators play a key role in the theory of celebrated Calogero--Moser model \cite{calogero69}
(see also the reviews \cite{perelomov81}, \cite{polychronakos}). In particular, a complete set of quantum
integrals can be described in terms of Dunkl operators \cite{Heckman}. Modifications of the operators
\eqref{Dunkl} were  used in \cite{poly92, brink92} to define creation-annihilation operators and
quantum integrals for the rational Calogero-Moser model in the harmonic confinement.



The operators ${\mathcal D}_i$ together with the multiplication operators $x_i$ and the group algebra  $\C W$ with their standard action 
generate rational Cherednik algebra $H_g(W)$  \cite{EG}. We also refer to the books \cite{etingof, cher} for the rich theory of Cherednik algebras and their connections to other areas.

In this paper we are concerned with the deformation of the quantum angular momentum
generators, constructed by means of the Dunkl operators:
\begin{equation}
\label{Mkl'}
{\mathcal M}_{ij}=x_i {\mathcal D_j}-x_j {\mathcal D}_i.
\end{equation}
Our motivation still comes from Calogero--Moser systems, and, more specifically, angular part  of the rational quantum Calogero--Moser Hamiltonian. This system plays a role in the original Calogero's work \cite{calogero69}, and
it has been studied more recently in a number of papers both in the classical \cite{hln,hkl} and quantum \cite{feigin, flp} cases.
Thus the angular Hamiltonian  is shown to define a superintegrable
system on $N-1$ dimensional hypersphere (\cite{hkl}, see also \cite{woj83}), and a way to represent conserved charges was developed in \cite{hln}.  We note that identifying Liouville charges is still an open problem.
The operators \eqref{Mkl'} were already used in \cite{feigin} as building blocks for the intertwining operators between the angular Hamiltonians  with coupling constants different by an integer.

In this paper we firstly deal with the subalgebra of the rational Cherednik algebra of type $A_{N-1}$ generated by the elements ${\mathcal M}_{ij}$ and by the permutations $s_{ij}$. A close algebra already appeared in the work of V. Kuznetsov \cite{kuznetsov} in connection with the Calogero--Moser system.  We explain that the angular Calogero--Moser Hamiltonian can be realised as an element in this algebra. Furthermore, we show that it generates the centre of the algebra.
This operator can be considered as a deformation of the usual quadratic Casimir invariant of $so(N)$,
which  corresponds to the angular momentum square.

We also describe all the  relations which generators ${\mathcal M}_{ij}$ satisfy. The commutation relations are simple extensions of the usual $so(N)$  relations
with metrics replaced by pairwise particle permutations. There are additional ``crossing" relations which are quadratic in the generators too. We describe a basis in this algebra and show that it is a non-homogeneous quadratic algebra of Poincar\'e-Birkhoff-Witt (PBW) type in the sense of \cite{BG} (see also \cite{AS}). These results are generalised for any Coxeter root system $\mathcal R$ in the final section.

We also consider the $gl(N)$ version of the subalgebra of the rational Cherednik algebra which is generated by the elements $x_k {\mathcal D}_l$ rather than their combinations ${\mathcal M}_{kl}$, and by the group algebra.
A closely related algebra in type $A_{N-1}$ was considered in \cite{tur94}. We describe relations in this algebra too and it gives another example of a non-homogeneous quadratic algebra of PBW type.

\section{Angular  Calogero-Moser Hamiltonian}
\label{sec:angular}
\vspace{5mm}

We will work with the gauged Dunkl operators defined by
\begin{equation}
\label{nabla}
\nabla_i=\p_i -  \sum_{\nad{j=1}{j\ne i}}^N \frac{g}{x_i-x_j} s_{ij}.
\end{equation}
They correspond to \eqref{Dunkl} via
the transformation $\psi\to \prod_{i<j}(x_i-x_j)^g\psi$
 of the wavefunctions. 
 Applying them instead of usual momenta operators for the free-particle system, we arrive
 at the modified Hamiltonian \cite{polychronakos}
 \begin{equation}
 \label{Hdunkl}
H=-\frac12\sum_{i=1}^N \nabla_i^2  = -\frac12\sum_{i=1}^N \partial_i^2 + \sum_{i<j}\frac{g(g-s_{ij})}{(x_i-x_j)^2}.
\end{equation}
 Let now $\text{Res}(A)\equiv \text{Res}_+(A)$ be the restriction of an ${\cal S}_N$-invariant operator $A$
 to the space of symmetric functions, and let $\text{Res}_{-}(A)$ be its restriction to the antisymmetric functions.
 Then
\begin{gather}
\label{heckmform}
\text{Res}_\pm\, H =  H_\pm,
\end{gather}
where  $H_\pm$ is the Calogero--Moser Hamiltonian \cite{Heckman}:
\begin{equation}
\label{Hpm_nonspherical}
H_\pm=-\frac12 \sum_{i=1}^N \p_i^2 +\sum_{i<j} \frac{ g (g\mp 1)}{(x_i-x_j)^2}.
\end{equation}

Apart from $g=0$ case, when the algebra generated by $\nabla_i, x_j$ reduces to the Heisenberg algebra,
the algebra formed by the coordinates and Dunkl operators for $g\ne0$
is not closed:
\begin{equation}
\label{com-nabla}
[\nabla_i,x_j]=S_{ij},
\qquad
[\nabla_i,\nabla_j]=0,
\qquad
[x_i,x_j]=0,
\end{equation}
where for the later convenience the pairwise permutation
operators are rescaled and the new notation $S_{ij}$
is introduced:
\begin{equation}
\label{Sij}
S_{ij}=
\begin{cases}
 - g s_{ij}, & \text{for $i\ne j$},
 \\
 1+g\sum_{k\ne i}s_{ik}=1-\sum_{k\ne i}S_{ik}, & \text{for $i=j$}.
\end{cases}
\end{equation}


Introduce spherical coordinates and let $H_\Omega$, $H_{\Omega, \pm}$ be the corresponding angular Calogero--Moser Hamiltonians obtained by the separation of radial and angular variables:
\begin{gather}
\label{Hspherical}
 H=-\frac{\partial_r^2}{2 }-\frac{N-1}{2r}\partial_r+\frac{H_\Omega}{r^2},
 \\
 \label{Hspherical-pm}
  H_\pm=-\frac{\partial_r^2}{2 }-\frac{N-1}{2r}\partial_r+\frac{H_{\Omega, \pm}}{r^2}.
 \end{gather}

Define Dunkl angular momentum operators
\begin{equation}
\label{Mkl''}
M_{kl}= x_k \nabla_l-x_l \nabla_k.
\end{equation}
 It appears that the operators
$M_{kl}$ allow to express the angular Hamiltonian in analogy
with the formula \eqref{heckmform} for the usual Calogero-Moser Hamiltonian.

\begin{prop}
\label{CMangD}
The angular Hamiltonians can be obtained as
\begin{align}
\label{Homega}
H_\Omega & = -\frac12 \mathbf{M}^2   + \frac12 S(S-N+2),
\\
\label{Hpm}
H_{\Omega, \pm} & =  {\rm Res}_\pm\, H_\Omega  = - \frac12 \mathbf{M}^2 + \gamma_\pm,
 \qquad \gamma_\pm = g N(N-1)\frac{gN(N-1)\pm 2(N-2)}{8},
\end{align}
where we define the Dunkl angular momentum square and the symmetric group algebra invariant  operators, respectively, as
\begin{gather}
\label{M2}
\mathbf{M}^2=\sum_{i<j} M_{ij}^2,
\\
\label{S}
S=\sum_{i<j} S_{ij}.
\end{gather}
\end{prop}

\begin{proof}
The proposition is a consequence of  the following relation
\begin{equation}
\label{M2x2}
\mathbf{M}^2=\mathbf{x}^2\boldsymbol{\nabla}^2-(\mathbf{x}\cdot\boldsymbol{\nabla})^2
   + (2S-N+2)(\mathbf{x}\cdot\boldsymbol{\nabla}),
\end{equation}
where we use the usual vector and scalar product notations.

Indeed, using  $\mathbf{x}^2=r^2$ and substituting
\begin{equation*}
\mathbf{x}\cdot\boldsymbol{\nabla} = \mathbf{x}\cdot\boldsymbol{\partial}+S=r\partial_r+S
\end{equation*}
into \eqref{M2x2}, we obtain
\begin{equation*}
\mathbf{M}^2=r^2\boldsymbol{\nabla}^2 -  r^2\partial_r^2 -(N-1)r\partial_r + S(S-N+2).
\end{equation*}
This relation together with \eqref{S} and \eqref{Hspherical} proves the desired
expression \eqref{Homega} for the spherical Hamiltonian.  Its restriction for bosons (fermions)
\eqref{Hpm}, corresponding to the pure Calogero model, is obtained from
$$
\text{Res}_\pm \,S = \mp \, g\frac{N(N-1)}{2}.
$$
In order to complete the proof, we have to justify the relation \eqref{M2x2}.
Substituting \eqref{Mkl''} into \eqref{M2} and using the commutation relations
\eqref{com-nabla}, we have:
\begin{equation}
\label{calculation-start}
\begin{split}
\mathbf{M}^2 &= \sum_{k,l} (x_k \nabla_l x_k \nabla_l-x_l\nabla_k x_k \nabla_l)
\\
&= \mathbf{x}^2\boldsymbol{\nabla}^2-(\mathbf{x}\cdot\boldsymbol{\nabla})^2
    +\sum_{k,l}  ( x_k S_{kl} \nabla_l+S_{kl} x_k\nabla_l-S_{kk} x_l\nabla_l )
\\
&= \mathbf{x}^2\boldsymbol{\nabla}^2-(\mathbf{x}\cdot\boldsymbol{\nabla})^2 + (2S-N)(\mathbf{x}\cdot\boldsymbol{\nabla})
    +\sum_{k, l}  ( x_k S_{kl} +S_{kl} x_k)\nabla_l  .
\end{split}
\end{equation}
The last term can be simplified further by rewriting  the definition \eqref{Sij} of $S_{ij}$
 in the generating function form
\begin{equation}
\label{Sx}
\sum_kS_{lk} x_k =  x_l +\sum_{k\ne l} S_{lk} (x_k-x_l).
\end{equation}
Its application 
leads to the following identity:
\begin{equation}
\label{Sx'}
\sum_k\{ S_{lk},x_k\}
= 2x_l+\sum_{k\ne l}  \{ S_{lk},x_l-x_k\}= 2x_l.
\end{equation}
By substituting \eqref{Sx'} into \eqref{calculation-start} we get the desired expression \eqref{M2x2}.
This completes the proof.
\end{proof}

We mention that a related statement in the case of Dunkl operators associated with the group $\Z_2^N$ is contained in \cite{Dunkl12}.

\section{Dunkl angular momenta algebra}
Recall that the creation-annihilation operators $a_i^{+}$, $a_i^{-}\equiv a_i$ for the Calogero--Moser rational model in the harmonic confinement \cite{brink92}, \cite{poly92}  are given by
\begin{equation}
\label{a-pm}
a_i^\pm=\frac{\sqrt{2}}{{2}}(x_i\mp \nabla_i).
\end{equation}
They satisfy the relations
\begin{equation}
\label{ai}
{[}a_i,a_j^+{]}=S_{ij},
\qquad
[a_i,a_j]=[a_i^+,a_j^+]=0
\end{equation}
and  are hermitian conjugate to each other.
This coincides with the relations \eqref{com-nabla} for the operators $x_i, \nabla_j$ via
the correspondence
\begin{equation}
\label{a-nabla}
a^+_i\leftrightarrow x_i,
\qquad
a_i\leftrightarrow \nabla_i.
\end{equation}
The commutation relations  between creation-annihilation operators and  
 group elements have the form
\begin{equation}
\label{com-sx}
S_{ij}a^\pm_i=a^\pm_jS_{ij} ,
\qquad
[S_{ij},a^\pm_k]=0
\qquad \text{with} \quad i \ne j \ne k \ne i.
\end{equation}
Recall also that
\begin{gather}
\label{com-ss}
S_{ij}S_{ik}=S_{jk}S_{ij},
\qquad [S_{ij},S_{kl}]=0,
\qquad S_{ij}^2=g^2,
\quad S_{ij}=S_{ji},
\end{gather}
where differently denoted indexes are assumed again to have different values, and it follows that
\begin{equation}
\label{com-sii}
S_{ij}S_{jj}=S_{ii}S_{ij}, \qquad
[S_{ij},S_{kk}]=0.
\end{equation}

The deformed angular momentum  generators \eqref{Mkl''} are expressed in terms of  the deformed
creation-annihilation operators \eqref{a-pm} as
\begin{equation}
\label{Mkl}
M_{kl}=a^+_ka_l-a_l^+ a_k. 
\end{equation}
They  inherit the standard (anti-)commutation
relations with the permutations operators:
\begin{equation}
\label{com-MS}
[S_{ij},M_{kl}]=0,
\qquad
\{S_{ij},M_{ij}\}=0,
\qquad
S_{ij}M_{ik}=M_{jk}S_{ij},
\end{equation}
provided that differently denoted indexes have different values.


When the deformation parameter $g$ vanishes,
\eqref{Mkl} corresponds to the standard representation
of the generators of the algebra $so(N)$ in terms of
bosonic creation-annihilation operators or vector fields. In this case, the corresponding operators $M_{ij}^0$ satisfy
the well known commutation relation
\begin{equation}
\label{son-class}
[M_{ij}^0,M_{kl}^0]=M_{il}^0\delta_{jk}+M_{jk}^0\delta_{il}-M_{ik}^0\delta_{lj}-M_{jl}^0\delta_{ik}.
\end{equation}
For nontrivial values of the deformation parameter, these generators, like $a^\pm_i$ themselves,
do not  form a Lie algebra  any more, since their commutators include the coefficients
dependent on the  permutation operators. 

We are interested in the {\it Dunkl angular momenta algebra} $\H_g so(N)$ which is generated by the operators
$M_{kl}$ and by the group algebra $\C {\cal S}_N$.
It can be considered as a subalgebra of the rational Cherednik algebra $H_g({\cal S}_N)$ via the homomorphism which maps generators $M_{kl} \to x_k  {\mathcal D}_l - x_l {\mathcal D}_k$.

%

To get the commutators
between the operators $M_{ij}$,  one just replaces the Euclidean metric $\delta_{ij}$ in \eqref{son-class} by the permutation operator $S_{ij}$.
This matches  the  deformation rules in the commutators given by \eqref{com-nabla}
or \eqref{ai}.
More precisely,  the following statement takes place whose version in another gauge is contained in \cite{kuznetsov} (Section 8,  arXiv version only).
\begin{prop} {(cf. \cite{kuznetsov})}
\label{prop1}
Operators  \eqref{Mkl} satisfy the following commutation relations:
\begin{equation}
\label{son}
[M_{ij},M_{kl}]=M_{il}S_{jk}+M_{jk}S_{il}-M_{ik}S_{lj}-M_{jl}S_{ik},
\end{equation}
where $1\le i,j,k,l \le N$.
\end{prop}

\begin{proof}
In order to prove \eqref{son}, firstly,  we
verify the commutators
\begin{equation*}
\label{sl}
[a_i^+a_j,a_k^+a_l]=a_i^+S_{jk}a_l-a_k^+S_{il}a_j,
\end{equation*}
which is  a simple consequence of the relations \eqref{ai}.
Then using them together with the definition \eqref{Mkl},
we get
\begin{equation}
\label{asym}
[M_{ij},M_{kl}]=\mathop{\text{asym}}_{[ij],[kl]}(a_i^+S_{jk}a_l-a_k^+S_{il}a_j)
=\mathop{\text{asym}}_{[ij],[kl]}(a_i^+S_{jk}a_l-a_l^+S_{jk}a_i),
\end{equation}
where in order to shorten the formulae, we introduce the operator, which
antisymmetrizes over the pairs of indexes included in square brackets:
\begin{equation}\label{asym-def}
\mathop{\text{asym}}_{[ij]}A_{ij}:=A_{ij}-A_{ji}.
\end{equation}

The simplest case is when the values of all four indices differ.
Then permutation operators commute with the annihilation-creation operators in all terms,
and we arrive at the desired commutation relation \eqref{son}.

It remains to consider the less trivial case when only three indexes differ.
It is sufficient to check $k=j$ case only. Then
\eqref{asym} acquires  the following form:
\begin{equation}
\label{asym-2}
\begin{split}
[M_{ij},&M_{jl}]=\mathop{\text{asym}}_{{[ij],[jl]}}(a_i^+S_{jj}a_l-a_l^+S_{jj}a_i)
\\
&=a_i^+S_{jj}a_l - a_j^+S_{ij}a_l - a_i^+S_{jl}a_j - a_l^+S_{jj}a_i + a_l^+S_{ij}a_j + a_j^+S_{jl}a_i
\end{split}
\end{equation}
Then moving permutations to the right in order to get the desirable result,
the additional  commutators appear in contrast to the previous case. However,
such unwanted terms are canceled out.
Indeed, the only nontrivial commutations between
permutations and creation-annihilation
operators are given by
\begin{gather}
\label{Sjjai}
[S_{jj}, a^\pm_i] = [S_{ij},a^\pm_j] = (a^\pm_i-a^\pm_j) S_{ij},
\\
\label{Sjjaj}
[S_{jj},a^\pm_j]=\sum_{k \ne j}(a^\pm_j-a^\pm_k)S_{kj},
\end{gather}
as can be deduced from the definition \eqref{Sij}.
Due to \eqref{Sjjai}, the appeared commutators eliminate each other:
$$
a_i^+\left([S_{jj},a_l] - [S_{jl},a_j]\right)
-a_l^+\left([S_{jj},a_i] - [S_{ij},a_j]\right)=0.
$$
Therefore, the commutation relation \eqref{son} holds for arbitrary index values,
and we have finished its proof.
\end{proof}

The relation \eqref{son} 
is reduced to the usual
commutation \eqref{son-class} for $so(N)$ generators in the nondeformed limit $g=0$.

The order of operators $S_{ij}$ and $M_{kl}$ in the right-hand side of the equation \eqref{son}
must be respected in general.
Using the relations \eqref{com-MS},  some permutations
in the right-hand side of the commutator \eqref{son} can be moved to the left of the Dunkl angular momenta generators.
In particular, if all the indexes differ then the operators
commute but, in general case,  the obtained in this way relations
will be more complicate.
However, it is easy to verify that all permutations can be moved
to the left simultaneously:
\begin{equation}
\label{son-rev}
[M_{ij},M_{kl}]=S_{jk}M_{il}+S_{il}M_{jk}-S_{lj}M_{ik}-S_{ik}M_{jl}.
\end{equation}
The above relation follows also from the  invariance of the subalgebra $\H_g so(N)$
under the Hermitian conjugation,
\begin{equation}
\label{herm}
 M_{ij}^+=-M_{ij}, \qquad S_{ij}^+=S_{ij},
\end{equation}
which is inherited from the bigger algebra $H_g({\cal S}_N)$.


We also note the equivariance of the relations \eqref{son}, \eqref{son-rev} under the symmetric group action. It is easy to verify using the relations
\eqref{com-MS},  \eqref{com-ss}, and \eqref{com-sii}
that any permutation $\sigma\in {\cal S}_N$, when acting on the commutator  \eqref{son}, just permutes
all its indexes: $i\to \sigma(i)$.

\section{\label{sec:crossing} Crossing relation and PBW property}
Let us establish another quadratic relation, which  the
generators of the  deformed angular momentum algebra $\H_g so(N)$ satisfy.
We call it the crossing relation.
\begin{prop}
The deformed angular momentum operators \eqref{Mkl} satisfy the following crossing relations:
\begin{equation}
\label{cros-quant}
M_{ij}M_{kl}+M_{jk}M_{il}+M_{ki}M_{jl}=M_{ij}S_{kl}+M_{jk}S_{il}+M_{ki}S_{jl},
\end{equation}
where $1\le i,j,k,l \le N$.
\end{prop}

\begin{proof}
Note that in both sides of the relation \eqref{cros-quant} the sum is taken over
the cyclic permutations of the first three indexes $i,j,k$, so that
it can be rewritten as
\begin{equation}
\label{cros-cyc}
\mathop{\text{cyclic}}_{(ijk)} M_{ij}M_{kl} = \mathop{\text{cyclic}}_{(ijk)} M_{ij}S_{kl},
\end{equation}
where for  convenience the cyclic permutation sum operator is introduced:
\begin{equation}
\label{def-cyc}
\mathop{\text{cyclic}}_{(ijk)} A_{ijk} := A_{ijk}+A_{jki}+A_{kij}.
\end{equation}
Note that the right-hand side of the equality \eqref{cros-quant} contains the terms, which appear also in the
right-hand side of the commutation relation \eqref{son}. 
It is easy to verify that if the values of any two of the indexes $i, j, k, l$
coincide, then the relation \eqref{cros-quant}  either reduces  to the commutation relation, or
its both side vanish trivially.
Therefore, it is enough to consider the case when the values of all four indexes differ.

Define the usual normal order $\mathcal{N}()$ with the creation
operators placed on the left hand side  and annihilation operators
on the {right}. For instance, $\mathcal{N}(a_i^+a_ja_k^+a_l)=a_i^+a_k^+a_ja_l$.
The operators commute under the normal order:
$\mathcal{N}([M_{ij},M_{kl}])=0$.
Since the classical momenta obey the crossing relations,
the same relation holds also for the normal order:
\begin{equation}
 \label{crossing}
\mathcal{N} (M_{ij}M_{kl}+M_{jk}M_{il}+M_{ki}M_{jl})=0.
\end{equation}
Using the definition \eqref{Mkl}, one can express each product
in terms of the normal order as
\begin{equation*}
M_{ij}M_{kl}=\mathcal{N}(M_{ij}M_{kl})
+\mathop{\text{asym}}_{[ij][kl]}a_i^+a_l S_{jk}.
\end{equation*}
Therefore, the right hand side of the crossing relation equals
\begin{equation*}
\begin{aligned}
M_{ij}M_{kl}+M_{jk}M_{il}+M_{ki}M_{jl}
&=a_i^+a_l S_{jk} + a_j^+a_k S_{il} - a_j^+a_l S_{ik} - a_i^+a_k S_{jl}
\\
&=a_j^+a_l S_{ki} + a_k^+a_i S_{jl} - a_k^+a_l S_{ji} - a_j^+a_i S_{kl}
\\
&=a_k^+a_l S_{ij} + a_i^+a_j S_{kl} - a_i^+a_l S_{kj} - a_k^+a_j S_{il}
\end{aligned}
\end{equation*}
Canceling out and grouping together similar terms, we arrive at the
desired expression \eqref{cros-quant}.
\end{proof}

Note that all the permutation operators in the right-hand side of the crossing relations
can be moved in front of the angular momentum operators:
\begin{equation}
\label{cros-cyc-2}
\mathop{\text{cyclic}}_{(ijk)} M_{ij}M_{kl}=\mathop{\text{cyclic}}_{(ijk)} S_{kl}M_{ij}.
\end{equation}
In the  the case of two equal indexes the relations
\eqref{cros-cyc} and \eqref{cros-cyc-2} are reduced to the
commutation relations \eqref{son} and \eqref{son-rev} respectively,
whose equivalence is already established.

Applying the hermitian conjugate  to the crossing relations \eqref{cros-quant}
and using \eqref{herm}, we obtain the equivalent relations
\begin{equation}
\label{cros-cyc-3}
\mathop{\text{cyclic}}_{(ijk)} M_{li}M_{jk}
=\mathop{\text{cyclic}}_{(ijk)} S_{li}M_{jk}
=\mathop{\text{cyclic}}_{(ijk)} M_{ij}S_{kl}
=\mathop{\text{cyclic}}_{(ijk)} M_{ij}M_{kl}.
\end{equation}
%
%

Like for the commutators, the system of crossing relations \eqref{cros-cyc-3} is invariant
with respect to the permutation group.

The sum of the equations  \eqref{cros-cyc} and \eqref{cros-cyc-2}
gives rise to a simpler crossing relation
written in terms of anticommutators with vanishing right-hand side:
\begin{equation}
 \label{cros-anticomm}
\mathop{\text{cyclic}}_{(ijk)}\, \{M_{ij},M_{kl}\}=0.
\end{equation}
This is the analogue of the crossing relations in quantum case, which possesses
the same symmetry as the classical one.
Its direct consequence is the vanishing of antisymmetrized product
\begin{equation}
 \label{cros-class}
\mathop{\text{asym}}_{[ijkl]} M_{ij}M_{kl}=0.
\end{equation}
For even values of $N$, this results in vanishing of the Pfaffian
\begin{equation}
 \label{casimir-N}
\sum_{i_1\dots i_N} \varepsilon_{i_1\dots i_N}
M_{i_1i_2}M_{i_3i_4}\ldots M_{i_{N-1}i_N}=0,
\end{equation}
which corresponds to a peculiar Casimir elements of $so(N)$ Lie algebra.
Here $\varepsilon_{i_1\dots i_N}$ is the Levi-Civita anti-symmetric tensor.

The crossing relation has  a clear graphical interpretation.
We represent the Dunkl angular momentum tensor $M_{ij}$ by an arrow with
tensor indexes at endpoints as shown on Figure~\ref{fig:Mij}.

The permutation operator $S_{ij}$ is depicted by a dashed bond.
Then the Figure~\ref{fig:crossing} describes schematically the
crossing relation \eqref{cros-quant}.  Each of six terms from the fist equation
is a product of two operators taken in a certain quantum order: the operator, containing
the index $l$  is positioned on the right side.
The ordering is not essential when we omit the first-order terms in momenta,  as shown
in the second equation on the Figure.

\psset{nodesep=2pt,dotsize=3pt,ncurv=0.8,ArrowInside=->,arrows=-,
           ArrowInsidePos=0.9,ArrowInsideNo=1,arrowscale=1.2}
\begin{figure}[t!]
\begin{pspicture}(-0.2,-0.4)(7.4,0.4)
    \def\i{1.3}\def\j{2.7}
    \pnode(\i,0){i} \pnode(\j,0){j}
    \psdots[linecolor=blue](\i,0)(\j,0)
    \ncline{i}{j}
    \uput[180](\i,0){$M_{ij}\;=\;\;$}
    \uput[-90](\i,0){$i$}
    \uput[-90](\j,0){$j$}
    \def\i{6}\def\j{7.4}
    \pnode(\i,0){i} \pnode(\j,0){j}
    \ncline[arrows=-,ArrowInside=-,linestyle=dashed]{i}{j}
    \uput[180](\i,0){$S_{ij}\;=\;\;$}
    \uput[-90](\i,0){$i$}
    \uput[-90](\j,0){$j$}
    \psdots[linecolor=blue](\i,0)(\j,0)
\end{pspicture}
\caption{\label{fig:Mij} Graphical representation for the deformed angular momentum and permutation
operators.}
\end{figure}

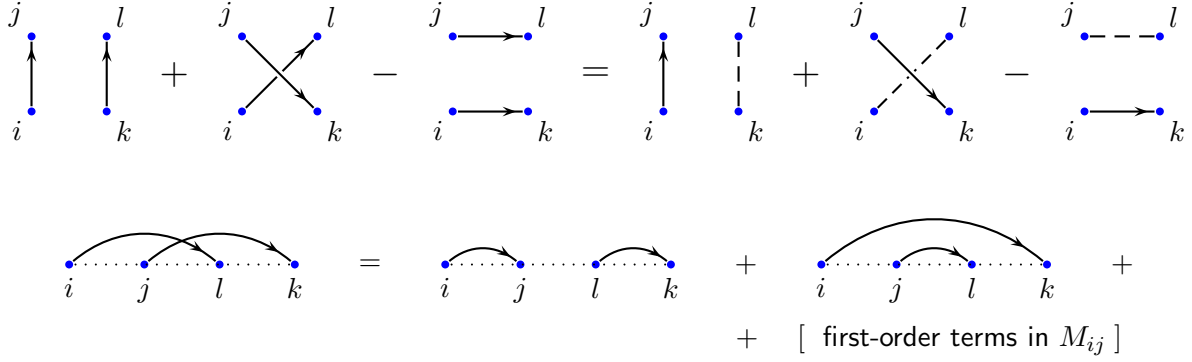
\begin{figure}[t!]
\begin{pspicture}(0,-0.5)(15,1.5)
    \plaquette{0}{0}{1}
    \ncline{k}{l}
    \ncline{i}{j}
    \plaquette{2.8}{0}{1}
    \ncline{i}{l}
    \ncline[linecolor=white,linewidth=3pt,arrows=-,ArrowInside=-]{j}{k}
    \ncline{j}{k}
    \plaquette{5.6}{0}{1}
    \ncline{j}{l}
    \ncline{i}{k}
    \plaquette{8.4}{0}{1}
    \ncline[arrows=-,ArrowInside=-,linestyle=dashed]{k}{l}
    \ncline{i}{j}
    \plaquette{11.2}{0}{1}
    \ncline[arrows=-,ArrowInside=-,linestyle=dashed]{i}{l}
    \ncline[linecolor=white,linewidth=3pt,arrows=-,ArrowInside=-]{j}{k}
    \ncline{j}{k}
    \plaquette{14}{0}{1}
    \ncline[arrows=-,ArrowInside=-,linestyle=dashed]{j}{l}
    \ncline{i}{k}
    \rput(1.9,0.5){\large$+$}
    \rput(4.7,0.5){\large$-$}
    \rput(7.5,0.5){\large$=$}
    \rput(10.3,0.5){\large$+$}
    \rput(13.1,0.5){\large$-$}
\end{pspicture}
\begin{pspicture}(1,-0.2)(15,2.5)
   \fourdotsijkl{1}{1}{1}
   \ncarc[arcangle=40]{x1}{x3}
   \ncarc[arcangle=40]{x2}{x4}
   \fourdotsijkl{6}{1}{1}
   \ncarc[arcangle=40]{x1}{x2}
   \ncarc[arcangle=40]{x3}{x4}
   \fourdotsijkl{11}{1}{1}
   \ncarc[arcangle=40]{x1}{x4}
   \ncarc[arcangle=40]{x2}{x3}
   \rput(5,1){$=$}\rput(10,1){$+$}\rput(15,1){$+$}
   \rput[r](15,0){$+\quad$ [\;\textsf{ first-order terms in $M_{ij}$}\;]}
\end{pspicture}
\caption{\label{fig:crossing}Graphical representation of the crossing relations \eqref{cros-quant}.
Each term is a products of two operators from Figure~\ref{fig:Mij}, with one with index $l$  positioned on the right.
The change of the operator order affects the right-hand side of the first equality only, and it is not essential for the second equality.}
\end{figure}

Now we are going to show that all other relations in the algebra $H_g^{so(N)}$ follow from the ones already established.
Let us consider an abstract associative algebra $\mathcal A$ over $\C$ generated by elements $M_{ij}$
and by the symmetric group algebra $\C {\cal S}_N$ such that the relations \eqref{com-MS}, \eqref{son},
\eqref{cros-quant} are imposed, as well as $M_{ji}=-M_{ij}$.
We are interested in a linear basis of this algebra.

Consider the monomials composed from deformed angular momenta
and pairwise permutations.
Then using the  commutation relations \eqref{com-MS}, the
permutations can be moved to the right giving rise to
\begin{equation}
\label{monom}
M_{i_1 j_1}^{n_1}\ldots M_{i_k j_k}^{n_k} \sigma
\qquad
\text{with}
\quad i_s<j_s,
\quad n_s > 0,
\quad k \ge 0,
\quad \sigma \in {\cal S}_N.
\end{equation}
The induction on $n=\sum_{s=1}^k n_s$  and the commutation relation \eqref{son}
ensure a rearrangement of the elements $M_{i_s j_s}$
in \eqref{monom} according to some predefined order. Thus we choose the indexes
to be ordered first by the values of $i_s$, then,  if they  equal, by the values of $j_s$:
 \begin{equation}
 \label{order}
 i_1\le \ldots \le i_k,
 \qquad \text{and} \qquad
 i_s=i_{s+1} \;\; \Rightarrow\;\;  j_s<j_{s+1}.
 \end{equation}
However, the ordered monomials are still not lineally independent.
 There is an additional restriction  imposed by the crossing relations
 \eqref{cros-quant}, which survives in the classical case too and can be easily formulated
 using the graphical representation.

 Let us mark points $1, 2,\ldots, N$ on the real line and connect the pairs $(i_s, j_s)$ by $n_s$
 directed semicircles in the upper half-plane.  We get in this way
 the graphical representation of the angular momentum part of the monomial \eqref{monom},
 see an example on Figure~\ref{fig:monom}.
Then the crossing bonds can be untangled by successive application of  the crossing relation,
which is represented on Figure~\ref{fig:crossing}.
Therefore, we arrive at the set  of monomials, which do not have intersecting semicircles,
that is they satisfy the condition
 \begin{equation}
 \label{non-cros}
 i_s < i_{s'} < j_s \quad\Rightarrow\quad
  j_{s'} \le j_s.
\end{equation}

\begin{figure}[t!]
\begin{pspicture}(-2.4,0)(8,3)
    \def\xa{2}\def\dx{2}\def\n{4}
    \uput[180](\xa,0){$M_{12}^2M_{14}M_{23}M_{24}^3M_{34}\;=\;\;$}
    \multido{\inode=1+1,\nd=\xa+\dx}{\n} {
        \psdot[linecolor=blue](\nd,0)
        \pnode(\nd,0){x\inode}
        \uput[-90](\nd,0){$\inode$}
      }
   \psset{nodesep=2pt,ncurv=0.8,ArrowInside=->,arrows=-,
           ArrowInsidePos=0.5,ArrowInsideNo=1,arrowscale=1.2}
  \ncarc[arcangle=20]{x1}{x2}
   \ncarc[arcangle=40]{x1}{x2}
   \ncarc[arcangle=20]{x2}{x3}
   \ncarc[arcangle=30]{x2}{x4}
   \ncarc[arcangle=40]{x2}{x4}
   \ncarc[arcangle=50]{x2}{x4}
   \ncarc[arcangle=20]{x3}{x4}
   \ncarc[arcangle=60]{x1}{x4}
   \ncline[linestyle=dotted,arrows=-,ArrowInsideNo=0]{x1}{x\n}
\end{pspicture}
\caption{\label{fig:monom} Graphical representation of a sample monomial, which does not contain intersecting
angular momentum bonds.}
\end{figure}
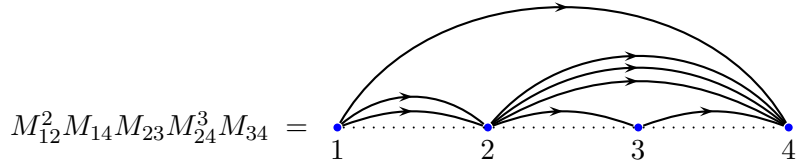

%

Let $V$ be a vector space of dimension $N$ with the basis $e_1,\ldots, e_N$. Let $\Lambda^2V$ be the second exterior power of $V$. Consider the tensor algebra $T(\Lambda^2V)$ and its smash product with the symmetric group algebra $\C {\cal S}_N$ where group elements act on $V$ and hence on $\Lambda^2 V$ by permuting the basis elements. We view operators $M_{ij}$ as elements of this smash product $T(\Lambda^2 V)\mathop{\#} \C {\cal S}_N$ via the mapping $M_{ij} \to (e_i \otimes e_j - e_j \otimes e_i) e$, where $e\in {\cal S}_N$ is the identity element.

\begin{thm}
\label{thmalg}
The deformation $\H_g so(N)$ is an associative algebra over $\C$ 
which is the quotient of the algebra $T(\Lambda^2 V)\mathop{\#} \C {\cal S}_N$ over the relations  \eqref{son}, \eqref{cros-quant}.
The set of monomials \eqref{monom} with the restrictions \eqref{order}, \eqref{non-cros}
gives a basis of the algebra $\H_g so(N)$ for  any  $g\in \C$.
\end{thm}
\begin{proof}
It is easy to see in the representation \eqref{Mkl''} that monomials \eqref{monom} with the restrictions \eqref{order} and \eqref{non-cros},
are linearly independent for $g=0$.
Moreover their classical version, when $M_{ij}$ is replaced with the classical angular momenta
\begin{equation}
\label{Mcl}
M_{ij}^\text{cl}=x_i p_j - x_j p_i,
\end{equation}
are linearly independent too (see e.g. \cite{hln}).
By taking the highest symbols,
considered as elements of the smash product algebra $\C[x,p]\mathop{\#}{\cal S}_N$,
it follows that these monomials are linearly independent for any $g$.
\end{proof}
Theorem \ref{thmalg} and its proof show that $H_g^{so(N)}$ is a flat family of nonhomogeneous quadratic
algebras over $\C {\cal S}_N$ of  Poincar\'e--Birkhoff--Witt type in the terminology  of \cite{BG} (see also \cite{EG}, \cite{AS}).

Note that the constructed non-intersection monomial basis is similar to the well-known overcomplete valence-bond
basis among the singlet states for usual quantum spins, introduced by Temperley and Lieb and
subjecting to the similar crossing relations (with the trivial right-hand  side)
\cite{TemperleyLieb}.
%


\section{ The centre}
\label{sec:centre}
Consider the square of the Dunkl angular momentum operator $\mathbf{M}^2$ given by \eqref{M2}, which is
an analogue  of the second order Casimir element
of the usual $so(N)$ algebra.
Evidently, it is invariant with respect to the permutations $S_{ij}$. However,
it commutes with all the elements $M_{ij}$ only in the limit $g=0$. In order to obtain
a true central element, the operator \eqref{M2} must be supplemented by an element from the
symmetric group algebra. Such an operator appeared in \cite{kuznetsov}. We show that this operator generates the whole centre. Based on Proposition \ref{CMangD} this element can be identified with the angular Calogero--Moser Hamiltonian.


\begin{thm}\label{soncentre}
\label{casthm}
The angular Hamiltonian \eqref{Homega} is invariant with respect
to the  algebra $\H_g so(N)$, that is it belongs to its centre $Z$:
\begin{equation}
\label{com-HomegaM}
[H_\Omega,M_{ij}]=0, \qquad [H_\Omega,S_{ij}]=0.
\end{equation}
 Furthermore, $Z$ is generated by $H_\Omega$ and constants.
\end{thm}

\begin{proof}
The second equation in \eqref{com-HomegaM} reflects the fact that $S$
lies in the center of the symmetric group algebra: $[S,S_{ij}]=0$.

First we prove that the original Hamiltonian is invariant with respect to the algebra $\H_g so(N)$.
Since its invariance with respect to the particle permutations is evident, it remains
to verify
\begin{equation}
\label{com-HM}
[H,M_{ij}]=[\boldsymbol{\nabla}^2,M_{ij}]=0.
\end{equation}
Indeed, as is argued in the proof of Proposition~\ref{CMangD}, the relation
\eqref{Sx'} remains valid upon the substitution $x_i\to\nabla_i$ :
$$
\sum_k\{ S_{ik},\nabla_k\}= 2\nabla_i.
$$
As a consequence,  we obtain the relation \eqref{com-HM}
using also the commutations \eqref{com-nabla}:
\begin{equation}
\label{nabla2-M}
{[}\boldsymbol{\nabla}^2,M_{ij}]=
\mathop{\text{asym}}_{{[ij]}}\sum_k[ \nabla_k^2,x_i\nabla_j]
=\mathop{\text{asym}}_{{[ij]}}\sum_k\{S_{ik} \nabla_j,\nabla_k\}
=2\mathop{\text{asym}}_{{[ij]}}\nabla_i\nabla_j =0,
\end{equation}
where we have used the index antisymmetrization
\eqref{asym-def} for the convenience.

Then we note that the Dunkl angular momentum operators \eqref{Mkl''}, \eqref{nabla}
depend on angular coordinates only.  This becomes  evident
after the application of the involutive antiautomorphism  of the Cherednik algebra
with the mutual exchange of  $x_i$ and  $\nabla_i$,  to  the equation \eqref{com-HM}:
\begin{equation}
\label{x2-M}
[\mathbf{x}^2,M_{ij}]=[r^2,M_{ij}]=0.
\end{equation}
Together with \eqref{com-HM} and the expression of the Hamiltonian in
spherical coordinates  \eqref{Hspherical}, it
implies that  the angular Hamiltonian \eqref{Hdunkl}  preserves the Dunkl angular momentum,
i.e. the first relation in \eqref{com-HomegaM}.

Consider now an arbitrary element $B \in Z$. Let $B_0\in \C[x,p]\mathop{\#}{\cal S}_N$ be its highest symbol $Sym(B)$ in the smash product algebra of polynomials in coordinates $x_1,\ldots, x_N$ and momenta $p_1, \ldots, p_N$, and the symmetric group. Let $d\ge 0$ be its degree in $x$ or $p$ variables. Note that $B_0$ does not contain non-trivial group elements. Indeed, consider the decomposition of $B$ with regard to the basis as in Theorem \ref{thmalg}. Suppose that a term with $\sigma\in {\cal S}_N$, $\sigma\ne 1$ arises in the highest degree. By taking $M_{ij}$ such that $\sigma(M_{ij})\ne M_{ij}$ we get that
$$
\deg Sym[B, M_{ij}]=d+1,
$$
so the commutator cannot be zero.

Consider now the commutator $[B, M_{ij}]$ for an arbitrary $M_{ij}$. Note that the terms
in its symbol of degree $d$, which do not contain non-trivial group elements, coincide with the  classical Poisson commutator of $B_0$ and the usual
classical angular momentum  \eqref{Mcl}, 
which is zero.
The Poisson centre of the algebra generated by  the classical angular momenta $M_{ij}^\text{cl}$ is generated by $\sum_{i<j} (M_{ij}^\text{cl})^2$ (this can be derived, for instance, from the results of \cite{Weyl}). The rest follows by subtracting the corresponding power of $H_\Omega$ from $B$ and by inductive reducing the degree of $B$.
\end{proof}

%
%
%
%

\section{Example: deformation of $so(3)$ algebra}
Consider now the simplest case of the  $H_g^{so(3)}$ algebra, which describes a deformation of
the usual three-dimensional quantum angular momentum operators. Define 
\begin{gather*}
\label{Mi}
M_1=M_{23}, \qquad M_2=M_{31}, \qquad M_3=M_{12},
\\
\label{Si}
S_1=S_{23}, \qquad S_2=S_{31}, \qquad S_3=S_{12}.
\end{gather*}
Then the commutation relations \eqref{com-ss} and \eqref{com-MS} between the generators
take the form
\begin{gather*}
\label{com-ss-3}
S_1S_2 = S_3S_1, \qquad  S_1^2=g^2,
\\
\label{com-MS-3}
 \{S_1,M_1\}=0, \qquad S_1M_2=-M_3S_1,  \qquad S_1M_3=-M_2S_1
\end{gather*}
and their cyclic permutations.

The commutation relations \eqref{son} are equivalent to
\begin{equation*}
\label{so3}
 [M_1,M_2]=-M_3+(M_3-M_2)S_1+(M_3-M_1)S_2
\end{equation*}
and other relations are obtained by the cyclic permutations of three indexes. They can be rewritten in
a more symmetric form,
\begin{equation*}
\label{so3-2}
 [M_1,M_2]=-M_3+\{M_3,S_1+S_2\}.
\end{equation*}

The Casimir element, which is proportional to the angular Hamiltonian,
 has the following form in terms of introduced operators:
\begin{equation*}
H_\Omega \sim {M_1}^2+{M_2}^2+{M_3}^2 - S(S-1), 
\qquad
S=S_1+S_2+S_3.
\end{equation*}
It extends the spin square operator.

Note that for $so(3)$ case the crossing relations are either trivial
or reduce to the commutation relations.

\section{$gl(N)$ case}
\label{secg}

The algebra $\H_g so(N)$ can be included into a bigger subalgebra $\Hggl(N)$ of the rational Cherednik algebra.
We define it to be generated by the operators
\begin{equation}
\label{Ekl}
E_{kl}=a^+_k a_l 
\end{equation}
and by the group algebra $\C {\cal S}_N$ ($1\le k,l\le N)$.  The operators can be considered acting on meromorphic functions. We also have  $M_{kl}=E_{kl}-E_{lk}$.

The algebra can be realised as a subalgebra of the rational Cherednik algebra $H_g({\cal S}_N)$ via a conjugation which maps $\nabla_k \to {\mathcal D}_k$.
Further, there are more general but equivalent choices of the generators $\widetilde E_{kl} = (\alpha x_k + \beta \nabla_k)(\gamma x_l + \delta \nabla_l)$, where parameters $\a, \beta, \gamma, \delta$ are such that $\a \delta- \beta \gamma  \ne 0$.  They define  isomorphic algebras. Indeed,  the operators $\alpha x_k + \beta \nabla_k$ with different indecies $k$ pairwise  commute. Therefore these  generators correspond to different choices of generators of the rational Cherednik algebra. In particular, the algebra $\Hggl(N)$ is isomorphic to the subalgebra of the rational
 Cherednik algebra $H_g({\cal S}_N)$ which is generated by the elements $x_k {\mathcal D}_l$ and by the group algebra $\C{\mathcal S}_N$.


%
%

Generators $E_{kl}$ obey the standard commutation relations with the permutation operators:
\begin{equation}
\label{com-ES}
[S_{ij},E_{kl}]=0,
\qquad
S_{ij} E_{ij} = E_{ji} S_{ij},
\qquad
S_{ij}E_{ik}=E_{jk}S_{ij},
\qquad
S_{ij}E_{ki}=E_{kj}S_{ij},
\end{equation}
where all indexes are pairwise different.

As in the case of  the usual $gl(N)$, the rising and lowering generators are Hermitian conjugate
of each other:
\begin{equation}
\label{herm-E}
 E_{ij}^+=E_{ji}, \qquad S_{ij}^+=S_{ij}.
\end{equation}

We start from the analogue of the crossing relation \eqref{cros-quant}, which in this case
takes the following form:
\beq{crosgl}
E_{ij} E_{kl}- E_{il} E_{kj}= E_{il}S_{kj}-E_{ij}S_{kl}.
\eeq
For different values of $i,j,k$ this relation follows immediately from \eqref{Ekl} and
\eqref{ai}, \eqref{com-sx}. The nontrivial case is when $j=k\ne l$, for which we have:
\[
\begin{split}
E_{ij} E_{kl}- E_{il} E_{kj}&=a_i^+S_{jj}a_l - a_i^+S_{lj}a_j = E_{il}S_{jj}+(E_{il}-E_{ij})S_{lj}- E_{il}S_{lj}
\\
&=E_{il}S_{jj}-E_{ij}S_{lj},
\end{split}
\]
where in the second equality above the relation \eqref{Sjjai} was applied.

Taking the Hermitian conjugate of \eqref{crosgl}, we  get an equivalent relation,
with antisymmetrization over the first indexes of $E_{ij}$ generators. Both relations
can be written down in the following compact form:
\begin{gather}
\label{crosgl2}
\mathop{\text{asym}}_{[jl]} E_{ij} (E_{kl}+S_{kl})=\mathop{\text{asym}}_{[ik]} (E_{ij} +S_{ij})E_{kl}=0.
\end{gather}
Their combination with suitably chosen indexes results in the commutation relation
among the deformed $gl(N)$ generators
\begin{eqnarray}
\label{relgln1}
[E_{ij}, E_{kl}]= E_{il}S_{jk} - S_{il} E_{kj} + [S_{kl}, E_{ij}],
\end{eqnarray}
which holds for any index values. Its antisymmetrisation over
$i,j$ and $k,l$ immediately gives the commutation relations \eqref{son} for
deformed angular momenta.
When all indexes differ, the commutator  in the right-hand side of \eqref{relgln1}
disappears, and we arrive at a natural extension of $gl(N)$ commutation relations with metrics
tensor $\delta_{ij}$ replaced by the pairwise permutation operator $S_{ij}$, as for the $so(N)$ case considered above.
Some other commutation relations are less straightfowrward. Below we write down all particular cases of
the commutation relation \eqref{relgln1} provided that the values of all four indexes differ pairwise.
\begin{equation}
\label{relgln2}
\begin{gathered}
{[}E_{ij}, E_{kl}]= E_{il}S_{jk} - E_{kj}S_{il},
\qquad
[E_{ii}, E_{kl}]=E_{il}S_{ik}-E_{kl}S_{il},
\\
[E_{ii}, E_{jj}]=(E_{ii}-E_{jj})S_{ij},
\qquad
[E_{ij}, E_{jl}]=E_{il}S_{jj}+(E_{il}-E_{ij})S_{jl}-E_{jj}S_{il},
\\
[E_{ii}, E_{ij}]=E_{ij}S_{ii}-E_{ii}S_{ij},
\qquad
[E_{ij}, E_{kj}]=E_{ik}S_{jk}-E_{ki}S_{ij},
\\
[E_{ij}, E_{ji}]=E_{ii}S_{jj}-E_{jj}S_{ii}+(E_{ii}-E_{jj}-E_{ij}+E_{ji})S_{ij},
\end{gathered}
\end{equation}
and the equalities from the third line in \eqref{relgln2} must also be supplemented by their conjugate relations.
As for  $\H_g so(N)$  case, considered above, any monomial in $E_{ij}$ and $S_{ij}$ can be
expressed in the following form by moving all permutations to the right hand side:
\beq{basgln}
E_{i_1 j_1}^{n_1}\ldots E_{i_k j_k}^{n_k} \sigma, \qquad \sigma\in \mathcal{S}_N,
\eeq
where $k\ge 0, n_1, \ldots, n_k >0$.
Due to relations \eqref{crosgl}--\eqref{relgln1}, any such monomial can be linearly expressed in
terms of the monomials, where both index sets  $\{i_s\}$ and $\{j_s\}$ are ordered
 in some way, for example
 \begin{equation}
\label{sort}
i_1\le \ldots\le i_k
\qquad\text{and}\qquad
j_1\le \ldots\le j_k.
\end{equation}

Using the Wicks's theorem, one can decompose \eqref{basgln} as was described in Section~\ref{sec:crossing}.
In the  example below, only the highest-order term 
 in this normal ordering
decomposition is shown
$$
E_{11}^{n_1}E_{12}^{n_2}E_{22}^{n_3}E_{32}^{n_4}E_{33}^{n_5}E_{34}^{n_6}
\;=\;(a_1^+)^{n_1+n_2}(a_2^+)^{n_3}(a_3^+)^{n_4+n_5+n_6}a_1^{n_1}a_2^{n_2+n_3+n_4}a_3^{n_5}a_4^{n_6}\;\;+\;\; \ldots.
$$
Here the lower-order terms   in  $a_i^\pm$  are indicated by the dots.

The monomials \eqref{basgln} with the restriction \eqref{sort} are linearly independent, since
it is easy to see that there is a one-to-one correspondence between them and their ``highest symbols"
$$
a^{+\,n_1}_{i_1}\ldots a^{+\,n_k}_{i_k}\; a_{j_1}^{n_1}\ldots a_{j_k}^{n_k},
$$
and because of the PBW theorem for the rational Cherednik algebra \cite{EG}.

%

Let $V$ be a vector space of dimension $N$ with the basis $e_1,\ldots, e_N$. Consider the tensor algebra $T(V\otimes V)$ and its smash product with the symmetric group algebra $\C {\cal S}_N$ where group elements act on $V$ and hence on $V \otimes V$ by permuting the basis elements. We view operators $E_{ij}$ as elements of this smash product $T(V\otimes V)\mathop{\#} \C {\cal S}_N$ via the mapping $E_{ij} \to (e_i \otimes e_j) e$, where $e\in {\cal S}_N$ is the identity element.

Similarly to the case of $\H_g so(N)$ we have the following statement which implies that $\Hggl(N)$ is a flat family of nonhomogeneous quadratic algebras of  Poincar\'e--Birkhoff--Witt type.
\begin{thm}
The deformation $\Hggl(N)$ is an associative algebra over $\C$ 
which is the quotient of the algebra $T(V\otimes V)\mathop{\#} \C {\cal S}_N$ over the relations  \eqref{crosgl2}.
The set of monomials \eqref{basgln} with the restriction \eqref{sort}
gives a basis of the algebra $\Hggl(N)$ for  any  $g\in \C$.
\end{thm}

From the commutation relations \eqref{relgln1} it is easy to deduce that the element
\begin{equation}
\label{rho}
\rho=\sum_i E_{ii} - S = \sum_i a^+_ia_i-S,
\end{equation}
$S=\sum_{i<j} S_{ij}$, commutes with the entire algebra $\Hggl(N)$. It describes the Calogero Hamiltonian \eqref{Hpm}
in confined oscillator potential \cite{calogero69}, which can be seen from \eqref{Hdunkl} and \eqref{a-pm}
\cite{brink92}:
\begin{equation}
\begin{gathered}
\rho = H +\frac12\mathbf{x}^2-\frac{N}{2},
\\
\text{Res}_\pm\rho= H_\pm+\frac12\mathbf{x}^2-\frac{N}{2}.
\quad
\end{gathered}
\end{equation}
Therefore, the algebra  $\Hggl(N)$ describes the invariants of the Calogero-Moser Hamiltonian
$$
\frac12\sum_{i=1}^N (-\partial_i^2 + x_i^2) + \sum_{i<j} \frac{g(g-s_{ij})}{(x_i-x_j)^2},
$$
extended out of the space of identical particles.

Finally we note that $\rho$  and constants generate the centre of $\Hggl(N)$. This follows from the following lemma similarly to the proof of Theorem \ref{soncentre}.
\begin{lemma}
Consider the classical Poisson algebra generated by $E_{ij}^{cl}=x_i p_j$ ($1\le i,j\le N$) with the standard Poisson bracket. Then $C=\sum_{i=1}^N x_i p_i$ generates the centre $Z$ of the algebra.
\end{lemma}
\begin{proof}
Let $f(x,p)\in Z$. Since the Poisson bracket
\begin{equation}\label{clcom}
\{x_i p_j, f\} = (x_i \p_{x_j} - p_j \p_{p_i}) f=0
\end{equation}
it follows for $i=j$ that $f$ is a function of the variables $y_i=x_i p_i$ $(1\le i \le N)$. In these variables, the relation  \eqref{clcom} implies that $(\p_{y_i}-\p_{y_j})f=0$ so that $f$ is a function of $C=\sum y_i$ as stated.
\end{proof}


\section{Coxeter groups generalisations}

The previous results can be generalised to the case when the symmetric group ${\cal S}_N$ is replaced with a finite reflection group.

 Let $\mathcal R$ be a Coxeter root system in the Euclidean space $V \cong \R^N$ with the inner product denoted as $(\cdot, \cdot)$  (see \cite{Hum}). Let $e_1, \ldots, e_N$ be the standard basis. Denote by $s_\a$  the orthogonal reflection with respect to the hyperplane $(\a, x)=0$:
  $$
  s_\alpha x= x-\frac{2(\alpha,x)}{(\alpha,\alpha)}\alpha,
  $$
where  $x=\sum x_i e_i \in V$.   Let $W$ be the corresponding finite Coxeter group generated by reflections $s_\a$, $\a\in \mathcal R$. Let $g: {\mathcal R}\to \C$ be a $W$-invariant function, denote $g_\alpha = g(\alpha)$.

 The Dunkl operators take the form
\begin{equation}
\label{dunklgen}
\nabla_\xi = \partial_\xi - \sum_{\a \in {\mathcal R}_+} \frac{g_\a (\alpha, \xi)}{(\a,x)}s_\a,
\end{equation}
where $\xi \in \R^N$, and  ${\mathcal R}_+$ is a positive half of the root system.
Given $\xi, \eta \in \R^N$ we define the element $S_{\xi \eta}$ of the group algebra $\C W$ by
\beq{sksieta}
S_{\xi \eta} = (\xi, \eta) + \sum_{\a\in {\mathcal R}_+} \frac{2 g_\a (\a, \xi)(\a, \eta)}{(\a,\a)} s_\a.
\eeq
This element comes from the commutation
$$
[\nabla_\xi, (x, \eta)]=  [\nabla_\eta, (x, \xi)] = S_{\xi \eta}.
$$
Define the Dunkl angular momentum operator
\begin{equation}
\label{Mxi-eta}
M_{\xi \eta} = (x,\xi) \nabla_\eta - (x, \eta) \nabla_{\xi}.
\end{equation}
Note that
 for any $w\in W$ and  $\xi,\eta \in \R^N$ the following relations hold:
\begin{equation}
\label{com-Mw}
\begin{gathered}
w M_{\xi \eta}= M_{w(\xi) w(\eta)}w,
\qquad
M_{\xi \eta} = \sum_{i,j} \xi_i \eta_j M_{e_ie_j},
\qquad
M_{\xi \eta} = - M_{\eta \xi}.
\end{gathered}
\end{equation}
The {\it Dunkl angular momenta algebra} $H_g^{so(N)}(W)$  is defined to be  generated by the operators $M_{\xi \eta}$, and by the group algebra $\C W$. It can be considered as a subalgebra of the rational Cherednik algebra $H_g(W)$ via the homomorphism which maps $M_{e_i e_j} \to x_i {\mathcal D}_j - x_j {\mathcal D}_i$.
\begin{prop}
The Dunkl angular momenta satisfy
\beq{amcoxcom}
[M_{\xi \eta}, M_{\varphi \psi}] = M_{\xi \psi} S_{\eta \varphi} +M_{\eta \varphi} S_{\xi \psi} - M_{\xi \varphi} S_{\eta \psi} - M_{\eta \psi} S_{\xi \varphi},
\eeq
\beq{amcoxcross}
M_{\xi \eta} M_{\varphi \psi} + M_{\eta \varphi}  M_{\xi \psi} + M_{\varphi \xi} M_{\eta \psi}  = M_{\varphi \xi} S_{\eta \psi} +M_{\xi \eta} S_{\varphi \psi} + M_{\eta \varphi} S_{\xi \psi}
\eeq
for any $\xi, \eta, \varphi, \psi \in \R^N$.
\end{prop}
The proposition can be checked straightforwardly by making use of the formulas
$$
[S_{\varphi \eta},  \nabla_{\psi}] = [S_{\psi \eta},  \nabla_{\phi}],
\qquad
[S_{\varphi \eta},  (x, \psi)] = [S_{\psi \eta}, (x, \varphi)].
$$

Similarly to Theorem \ref{thmalg} for $W={\cal S}_N$, the algebra $H_g^{so(N)}(W)$ can be defined as the quotient of the smash product algebra $T(\Lambda^2V) \mathop{\#} W$. 
Let the element $M_{\xi \eta}$ correspond to $(x, \xi) \otimes (x, \eta)- (x, \eta) \otimes (x, \xi) \in \Lambda^2 V$. Then the
generating relations are given by \eqref{amcoxcom}, \eqref{amcoxcross}. The quotient has the basis given by monomials
\eqref{monom}, \eqref{order} as in Theorem \ref{thmalg},  where now $\sigma \in W$. Thus $H_g^{so(N)}(W)$ is a flat family of nonhomogeneous quadratic algebras over $\C W$ of  Poincar\'e--Birkhoff--Witt type.

Using the Dunkl operators \eqref{dunklgen}, a generalisation of the $\H_g so(N)$ invariant Hamiltonian \eqref{Hdunkl}
 can be defined as follows:
 \begin{equation}
 \label{HdunklW}
H=-\frac12\sum_{i=1}^N \nabla_{e_i}^2  =
-\frac12\sum_{i=1}^N \partial_{x_i}^2 + \sum_{\alpha\in {\mathcal R}_+}\frac{g_\alpha(g_\alpha-s_\alpha)}{2(\alpha,x)^2}.
\end{equation}
 Similarly to the case $W={\cal S}_N$ considered in  Section~\ref{sec:angular},
let $\text{Res} \, A$ be the restriction of a $W$-invariant element of the algebra $H_g^{so(N)}(W)$ to the space of $W$-invariant functions $\psi$:
 \begin{equation*}
\psi(s_\alpha x) = \psi(x) \qquad \forall \alpha \in {\mathcal R}.
\end{equation*}
 Then the Calogero-Moser Hamiltonian for a general Coxeter group $W$
 is
 \begin{equation}
 \label{Hdunkl-pmW}
H_+=\text{Res}\, H
=-\frac12 \sum_{i=1}^N \p_{x_i}^2 +\sum_{\a \in {\mathcal R}_+} \frac{ g_\a (g_\a-1)(\a,\a)}{2(\a,x)^2}.
\end{equation}
In the Hamiltonians \eqref{HdunklW}, \eqref{Hdunkl-pmW}, the radial and angular coordinates are separated giving rise to the
 relations  \eqref{Hspherical} and \eqref{Hspherical-pm}.  The angular Hamiltonians $H_\Omega$
and $H_{\Omega, +}$ can be expressed in terms of  the Dunkl angular
momentum operators \eqref{Mxi-eta}.
\begin{prop}
\label{angW}
The angular Calogero--Moser Hamiltonians \eqref{Hspherical} and \eqref{Hspherical-pm}
 can be obtained as
\begin{align}
\label{HomegaW}
H_\Omega & = -\frac12 \left(\mathbf{M}^2   - S(S-N+2)\right),
\\
H_{\Omega, +} & =  {\rm Res}\, H_\Omega,  
\nonumber
\end{align}
where we define the Dunkl angular momentum square and the symmetric group algebra invariant  operators,
respectively, as
\begin{equation*}
\label{CW}
\mathbf{M}^2=\sum_{i<j}^N M_{e_i e_j}^2,
\qquad
S=-\sum_{\a \in {\mathcal R}_+} g_\a s_\a.
\end{equation*}
\end{prop}

Generalisation of Theorem \ref{casthm} holds for any $W$ and it has the following form.
\begin{thm}\label{casthmW}
The centre $Z$ of the  algebra $H_g^{so(N)}(W)$ is generated by the angular Hamiltonian \eqref{HomegaW}
 and constants.
\end{thm}

\bigskip

The $gl(N)$ version of the algebra $H_g^{gl(N)}(W)$ can also be defined for any $W$ as the algebra generated by
$$
E_{\eta \xi}= (x, \eta) \nabla_\xi
$$
and $\C W$, where $\xi, \eta \in V$ and the Dunkl operators are given by \eqref{dunklgen}.
 It can also be thought as a subalgebra of the rational Cherednik algebra $H_g(W)$, 
 which is isomorphic to the quotient of the smash product algebra $T(V\otimes V) \mathop{\#} \C W$ over the relations  \eqref{crosgl}, \eqref{relgln1}, 
 where the elements $S_{ij}=S_{e_i e_j}\in \C W$ are given by \eqref{sksieta}. 
 The set of monomials \eqref{basgln}, \eqref{sort}  with $\sigma \in W$ gives a basis of the algebra 
 for  any $W$-invariant multiplicity function $g$, and this is also a quadratic algebra over $\C W$ of PBW type.
The centre is generated by constants and
$$
\rho = \sum_{i=1}^N E_{e_i, e_i} +\sum_{\a\in{\mathcal R}_+} g_\a s_\a.
$$
In the representation of the algebra $E_{\xi \eta}\to \widetilde E_{\xi \eta} = \frac12((x,\xi)-\nabla_\xi)((x,\eta)+\nabla_\eta)$ the central generator takes the form of the Calogero-Moser operator in the harmonic confinement associated with $W$:
$$
\rho \to H +\frac12\mathbf{x}^2 -\frac{N}{2},
$$
where $H$ is given by \eqref{HdunklW}.

\section{Concluding remarks}

The considered algebras $H_g^{so(N)}(W)$, $H_g^{gl(N)}(W)$  have the form of the quotients of smash product algebras $T(U)\mathop{\#} \C W$, where $U=\Lambda^2V$ or $U=V\otimes V$, and $V$ is  the reflection representation of the Coxeter group $W$. In addition to the commutation relations the generators satisfy extra quadratic relations. These extra relations are needed in order to have algebras with the PBW property. It also means that we consider  the semidirect products of the quotients of the corresponding universal enveloping algebras $U(so(N))$, $U(gl(N))$ and $W$ rather than the products of the universal enveloping algebras themselves.

Further, for the associated graded algebras the right-hand sides of the extra quadratic relations \eqref{amcoxcross}, \eqref{crosgl} vanish. Thus we are dealing with Pl\"ucker relations and the affine cone $\cal C$ over the grassmanian of two-planes in the $so(N)$ case and  the space  $\cal M$ of matrices of rank at most one in the $gl(N)$ case. It might be interesting to study the algebras $H_g^{so(N)}(W)$, $H_g^{gl(N)}(W)$ and their representations further, as well as their $t=0$ versions in the notation of \cite{EG}, in particular, in connection to (Poisson) geometry and singularities of the quotient of the spaces $\cal C$, $\cal M$ over the natural action of the Coxeter group $W$. It would also be interesting to see whether there are algebras with good properties analogously associated to some other partial flag varieties or unions of coadjoint orbits.

A sheaf of algebras was associated to varieties with finite group actions in \cite{Eti}. In the case of the projective space with the natural action of the Coxeter group $W$  the algebra of global sections is isomorphic to a quotient of the algebra $H_g^{gl(N)}(W)$ over the ideal generated by a central element \cite{Eti} (see also \cite{BM}). Affinity of this sheaf of algebras was investigated in \cite{BM}. Another interesting question is about the extension of these results to the algebra $H_g^{so(N)}(W)$ and the sheaf of algebras associated with the subspace of isotropic lines in the projective space.


We also note that integrability of the angular Calogero--Moser systems deserves further analysis as, in particular, a complete set of Liouville  charges remains unknown. The problem of Liouville integrability was one of the main motivations for this work and we hope that the work is a useful step towards its solution. There are a few reasons why the problem is difficult and remains unsolved which we try to explain (see also \cite{hln}). Firstly, in comparison with the case of the Calogero--Moser systems in the linear space, commutativity of homogeneous quantum integrals for the latter problem is quite straightforward and can be derived from the property that the highest term of any quantum integral should be polynomial \cite{Ch} .

Secondly, one  way to try to establish Liouville integrability for the angular Hamiltonian $H_\Omega$ is to consider the natural embedding of the universal enveloping algebras
$$
U(so(2))\subset U(so(3))\subset \ldots \subset U(so(N)).
$$
 This chain immediately gives a complete set of commuting quantum integrals for the Laplace-Beltrami operator on the sphere which corresponds to $g=0$ case. The corresponding operators are the realisations of the Casimir elements ${\bf M}^2$. However, it is not clear how to extend this observation to $g\ne 0$ as the corresponding angular momentum operators $M_{ij}$ start depend on the parameter $n$ of the algebra $so(n)$ via the summation term of the Dunkl operator. So the algebras $H_g^{so(n)} (W)$ for different $n$ do not naturally embed one into another already for $W={\mathcal S}_n$.

Nonetheless, in the $gl(N)$ case a version of this construction exists at least for classical Coxeter groups $W$. In this case a correction $E_{ii}'$ of the operator $E_{ii}$ by group algebra elements is known so that $[E_{ii}', E_{jj}']=0$ \cite{DO}. It would be interesting to try to extend this construction to other finite reflection groups $W$ and to try to modify Casimir elements ${\bf M}^2$ to a commuting family. It is not plausible though that modification purely by group algebra elements is possible in the $so(N)$ case. A more general algebra at $g=0$ one can alternatively start with is the quantum shift of argument algebra (see \cite{R}) in which case one allows operators of order higher than 2.

In the case of $gl(N)$ and classical Coxeter groups $W$  it follows from \cite{poly92} (for $W=A_{N-1}$) and  \cite{F} (for $W=B_N, D_N$) that the algebra $H_g^{gl(N)}(W)$ contains a complete set of Liouville quantum integrals for the Calogero--Moser Hamiltonian in harmonic confinement.  Indeed, it is shown that a complete set of commuting quantum integrals is expressed via the combinations $a_i^+ a_i$. This suggests that the symmetry algebra $H_g^{so(N)}(W)$ might also be the right one for the angular Calogero--Moser Hamiltonian. We also note that the algebra and the above description of its centre has been recently explored in the study of Calogero--Moser deformation of the Coulomb problem and, in particular, for the construction of the generalisation of the Runge--Lenz vector \cite{Hak}.

\vspace{5mm}

\begin{acknowledgments}
\noindent
The authors are grateful to G. Bellamy, V. Dotsenko,  P. Etingof, E. Feigin,  D.~Karakhanyan, H.~Khudaverdian, O.~Lechtenfeld and A.~Nersessian
for stimulating discussions and useful comments.
The work was partially supported by the Royal Society/RFBR joint project JP101196/11-01-92612,
and by the VolkswagenStiftung  under Contract no. 86 260.
T.H. is  supported also by  the Armenian State Committee of Science grant no. 15RF-039,
as well as by ANSEF grant no. 3501.
\end{acknowledgments}


\begin{thebibliography}{99}

\bibitem{dunkl89}
C.F. Dunkl,
\emph{Differential-difference operators associated to reflection groups}, \\
\href{http://www.ams.org/journals/tran/1989-311-01/S0002-9947-1989-0951883-8}%
{Trans. Amer. Math. Soc. {\bf 311} (1989) 167}.

\bibitem{calogero69}
F.~Calogero,
\emph{Solution of a three-body problem in one-dimension},\\
\href{http://dx.doi.org/10.1063/1.1664820}{J. Math. Phys. {\bf 10} (1969) 2191};
\emph{Solution of the one-dimensional N-body problems with quadratic and/or inversely quadratic
pair potentials},
\href{http://dx.doi.org/10.1063/1.1665604}{{\sl ibid.} {\bf 12} (1971) 419}.

\bibitem{perelomov81}
M.A.~Olshanetsky and A.M.~Perelomov,
\emph{Classical integrable finite-dimensional systems related to Lie algebras},
\href{http://dx.doi.org/10.1016/0370-1573(81)90023-5}{Phys. Rept. {\bf 71}  (1981) 313};
\emph{Quantum integrable systems related to Lie algebras},
\href{http://dx.doi.org/10.1016/0370-1573(83)90018-2}{Phys. Rept. {\bf 94}  (1983) 313}.

\bibitem{polychronakos}
A.P.~Polychronakos,
\emph{Physics and mathematics of Calogero particles},
J. Phys. A  {\bf 39} (2006) 12793,
\href{http://arxiv.org/abs/hep-th/0607033}{hep-th/0607033}.

\bibitem{Heckman}
G.J. Heckman,
\emph{A remark on the Dunkl differential-difference operators},\\
\href{http://dx.doi.org/10.1007/978-1-4612-0455-8_8}{Prog. in Math. {\bf 101} (1991) 181}.


\bibitem{poly92}
A. Polychronakos,
\emph{Exchange operator formalism for integrable systems of particles},
Phys. Rev. Lett. {\bf 69} (1992) 703,
\href{http://arxiv.org/abs/hep-th/9202057}{hep-th/9202057}.

\bibitem{brink92}
L. Brink, T. Hansson, and M. Vasiliev,
\emph{Explicit solution to the N-body Calogero problem},
Phys. Lett. B {\bf  286} (1992) 109,
\href{http://arxiv.org/abs/hep-th/9206049}{hep-th/9206049}.


\bibitem{EG}
P. Etingof and V. Ginzburg,
\emph{Symplectic reflection algebras, Calogero-Moser space, and deformed Harish-Chandra homomorphism},
\href{http://dx.doi.org/10.1007/s002220100171}{Invent. Math. {\bf 147} (2002) 243}.

\bibitem{etingof}
 P. Etingof,
 \emph{Calogero-Moser systems and representation theory}, European Mathematical Society, 2007,  
\href{http://arxiv.org/abs/math/0606233}{math/0606233}.

\bibitem{cher}
I. Cherednik,
{\it Double Affine Hecke Algebras}, CUP, (2005).

\bibitem{hln}
T.~Hakobyan, O.~Lechtenfeld, and A.~Nersessian,
\emph{The spherical sector of the Calogero model as a reduced matrix model},
Nucl. Phys. B  {\bf 858} (2012) 250,
\href{http://arxiv.org/abs/1110.5352}{arXiv:1110.5352}.

\bibitem{hkl}
\emph{The structure of invariants in conformal mechanics},
T. Hakobyan, D. Karakhanyan, and O.~Lechtenfeld,
\href{http://arxiv.org/abs/1402.2288}{arXiv:1402.2288};
T.~Hakobyan, S.~Krivonos, O.~Lechtenfeld and A.~Nersessian,
\emph{Hidden symmetries of integrable conformal mechanical systems},
Phys. Lett.  A {\bf 374} (2010) 801,
\href{http://arxiv.org/abs/0908.3290}{arXiv:0908.3290};
T.~Hakobyan, A.~Nersessian and V.~Yeghikyan,
\emph{Cuboctahedric Higgs oscillator from the Calogero model},
J.  Phys. A  {\bf 42}  (2009) 205206,
\href{http://arxiv.org/abs/0808.0430}{arXiv:0808.0430};
T.~Hakobyan, O.~Lechtenfeld, A.~Nersessian and A.~Saghatelian,
\emph{Invariants of the spherical sector in conformal mechanics},
J. Phys. A  {\bf 44} (2011) 055205,
\href{http://arxiv.org/abs/1008.2912}{arXiv:1008.2912}.


\bibitem{feigin}
M. Feigin,
\emph{Intertwining relations for the spherical parts of generalized Calogero operators},\\
\href{http://dx.doi.org/10.1023/A:1023231402145}
 {Theor. Math. Phys. {\bf 135} (2003) 497}.

\bibitem{flp}
M.~Feigin, O.~ Lechtenfeld, and A.~Polychronakos,
\emph{The quantum angular Calogero-Moser model},
JHEP {\bf 07} (2013) 162,
\href{http://arxiv.org/abs/1305.5841}{arXiv:1305.5841}.

\bibitem{woj83}
S.~Wojciechowski,
\emph{Superintegrability of the Calogero-Moser system},\\
\href{http://dx.doi.org/10.1016/0375-9601(83)90018-X}{Phys. Lett. A {\bf 95} (1983) 279}.

\bibitem{kuznetsov}
V.B.~Kuznetsov,
\emph{Hidden symmetry of the quantum Calogero-Moser system},
Phys. Lett. A {\bf 218} (1996) 212,
\href{http://arxiv.org/abs/solv-int/9509001}{arXiv:solv-int/9509001}.

\bibitem{BG}
A. Braverman and D. Gaitsgory,
\emph{Poincare-Birkhoff-Witt theorem for quadratic algebras of Koszul type},
\href{http://dx.doi.org/doi:10.1006/jabr.1996.0122}{J. Algebra {\bf 181} (1996)  315}.

\bibitem{AS}
A. Shepler and S. Witherspoon,
\emph{Poincare-Birkhoff-Witt theorems},
\href{http://arxiv.org/abs/1404.6497}{arXiv:1404.6497};
A. Shepler and S. Witherspoon,
\emph{A Poincare-Birkhoff-Witt theorem for quadratic algebras with group actions},
Trans. Amer. Math. Soc. {\bf 366} (2014) 6483,
\href{http://arxiv.org/abs/1209.5660}{arXiv:1209.5660};
A. Shepler and S. Witherspoon,
\emph{Drinfeld orbifold algebras},
 Pacific J. Math. 259 (2012) 161,
 \href{http://arxiv.org/abs/1111.7198}{arXiv:1111.7198};
G. Halbout, J.-M. Oudom, and X. Tang,
\emph{Deformations of orbifolds with noncommutative linear Poisson structures}
\href{http://dx.doi.org/doi:10.1093/imrn/rnq065}{Int. Math. Res. Not. 2011, no. 1, 1--39},
\href{http://arxiv.org/abs/0807.0027}{arXiv:0807.0027}.

\bibitem{tur94}
A. Turbiner,
\emph{Hidden algebra of the N-body Calogero problem},
Phys. Lett. B {\bf 320} (1994) 281,
\href{http://arxiv.org/abs/hep-th/9310125}{arXiv:hep-th/9310125}.

\bibitem{Dunkl12}
C. Dunkl,
\emph{Computing with differential-difference operators},\\
\href{http://dx.doi.org/doi:10.1006/jsco.1997.0341}{J. Symbolic Computation {\bf 28} (1999) 819};
\emph{Symmetric functions and $B_N$-invariant spherical harmonics}
J. Phys. A {\bf 35} (2002) 10391,
\href{http://arxiv.org/abs/math/0207122}{arXiv:math/0207122}.

\bibitem{TemperleyLieb}
N. Temperley and E. Lieb,
\emph{Relations between the ’percolation’ and ’colouring’ problem
and other graph-theoretical problems associated with regular planar lattices: some exact
results for the ’percolation’ problem},
\href{http://dx.doi.org/doi:10.1098/rspa.1971.0067}{Proc. Roy. Soc. Lond. A {\bf 322} (1971) 251}.

\bibitem{Hum}
J.E. Humphreys, {\em Reflection groups and Coxeter groups},
CUP, 1990.

\bibitem{Weyl}
H. Weyl,
\emph{The classical groups. Their invariants and representations}.
Princeton University Press, Princeton, NJ, 1997.

\bibitem{Eti}
P. Etingof,
\emph{Cherednik and Hecke algebras of varieties with a finite group action}, \\
\href{http://arxiv.org/abs/math/0406499}{arXiv:math/0406499}.

\bibitem{BM}
G. Bellamy and M. Martino,
\emph{Affinity of Cherednik algebras on projective space},
Algebra Number Theory {\bf 8} (2014) 1151,
\href{http://arxiv.org/abs/1305.2501}{arXiv:1305.2501}.


\bibitem{Ch}
O.A. Chalykh,
\emph{Additional integrals of the generalized quantum Calogero-Moser problem},
\href{http://dx.doi.org/10.1007/BF02069885}{Theor. Math. Phys. {\bf 109} (1996) 1269}.

\bibitem{DO}
C.F. Dunkl and E.M. Opdam,
\emph{Dunkl operators for complex reflection groups},\\
\href{http://dx.doi.org/doi:10.1112/S0024611502013825}{Proc. London Math. Soc. {\bf 86} (2003) 70}.

\bibitem{R}
L. Rybnikov,
\emph{The argument shift method and the Gaudin model},\\
\href{http://dx.doi.org/10.1007/s10688-006-0030-3}{Funct. Anal. Appl. {\bf 40} (2006)  188};
M. Nazarov and G. Olshanski,
\emph{Bethe subalgebras in twisted Yangians},
\href{http://dx.doi.org/10.1007/BF02099459}{Comm. Math. Phys. {\bf 178} (1996) 483}.

\bibitem{F}
M. Feigin,
\emph{Generalized Calogero–Moser systems from rational Cherednik algebras},\\
\href{http://dx.doi.org/10.1007/s00029-011-0074-y}{Selecta Math. (N.S.) {\bf 18} (2012) 253}.

\bibitem{Hak}
T. Hakobyan and A. Nersessian,
\emph{Runge-Lenz vector in the Calogero-Coulomb problem},
Phys. Rev. A {\bf 92} (2015) 022111,
\href{http://arxiv.org/abs/1504.00760}{arXiv:1504.00760};
T. Hakobyan, O. Lechtenfeld, and A. Nersessian,
\emph{Superintegrability of generalized Calogero
models with oscillator or Coulomb potential},
Phys. Rev. D {\bf 90} (2014)  101701(R),
\href{http://arxiv.org/abs/1409.8288}{arXiv:1409.8288}.



\end{thebibliography}
\end{document}